\documentclass[a4paper]{sig-alternate}
\usepackage[paper=letterpaper,width=7in,height=9.25in,centering]{geometry}
\usepackage[plainpages=false,pdfpagelabels,colorlinks=true,citecolor=blue,hypertexnames=false]{hyperref}
\usepackage{amssymb}
\usepackage{bbm}
\usepackage[utf8]{inputenc}
\usepackage{algo}

\def\<#1>{\langle#1\rangle}
\let\set\mathbbm
\def\vec#1{\mathbf{#1}}
\def\spread{\operatorname{Spread}}
\def\disp{\operatorname{Disp}}
\def\res{\operatorname{res}}

\newtheorem{theorem}{Theorem}
\newtheorem{proposition}{Proposition}
\newtheorem{corollary}{Corollary}
\newtheorem{lemma}{Lemma}
\newtheorem{definition}{Definition}
\newtheorem{example}{Example}

\newtheorem{algorithm}{Algorithm}

\conferenceinfo{ISSAC 2010,}{25--28 July 2010, Munich, Germany.}
\CopyrightYear{2010}
\crdata{978-1-4503-0150-3/10/0007}

\begin{document}

 \title{Partial Denominator Bounds for\\ Partial Linear Difference Equations}

\numberofauthors{2}
\author{
 \alignauthor Manuel Kauers\titlenote{Supported by the Austrian Science Fund
   (FWF) grants P20347-N18 and Y464-N18.}\\[\smallskipamount]
      \affaddr{RISC}\\
      \affaddr{Johannes Kepler University}\\
      \affaddr{4040 Linz (Austria)}\\[\smallskipamount]
      \email{mkauers@risc.jku.at}
 \alignauthor Carsten Schneider\titlenote{Supported by the Austrian Science Fund
   (FWF) grants P20162-N18 and P20347-N18.}\\[\smallskipamount]
      \affaddr{RISC}\\
      \affaddr{Johannes Kepler University}\\
      \affaddr{4040 Linz (Austria)}\\[\smallskipamount]
      \email{cschneid@risc.jku.at}
}

\maketitle

\begin{abstract}
  We investigate which polynomials can possibly occur as factors in the
  denominators of rational solutions of a given partial linear difference
  equation (PLDE).  Two kinds of polynomials are to be distinguished, we call
  them \emph{periodic} and \emph{aperiodic}.  The main result is a generalization
  of a well-known denominator bounding technique for univariate equations
  to PLDEs. This generalization is able to find all the
  aperiodic factors of the denominators for a given PLDE.
\end{abstract}




\section{Introduction}

Several algorithms in symbolic computation depend on a subroutine for finding
the rational solutions of an ordinary linear difference (or differential)
equation (OLDE), and several algorithms are known for implementing such a
subroutine~\cite{abramov71,Abramov:89b,abramov91,Petkov:92,hoeij98,weixlbaumer01,bostan06,chen08b,abramov10}.
On a conceptual level, the typical approach for finding rational solutions can
be divided into three steps.  In the first step, one constructs a polynomial~$Q$
such that the denominator~$q$ of any potential solution $p/q$ must divide~$Q$.
This polynomial $Q$ is called \emph{universal denominator} or \emph{denominator
  bound}.  In the second step, the universal denominator is used to transform
the given equation into a new equation such that $P$ is a \emph{polynomial}
solution of the new equation if and only if $P/Q$ is a \emph{rational} solution
of the original one.  In the third and final step, the polynomial solutions $P$
of the transformed equation are determined.

The first algorithm for computing universal denominators in the case of OLDEs
with polynomial coefficients was proposed in 1971 by Abramov~\cite{abramov71}  (see
Section~\ref{sec:abramov} below for a summary). It has been generalized to
$q$-difference equations~\cite{abramov95},
to matrix equations~\cite{barkatou99}, and also to
equations whose coefficients belong to domains other than polynomials. For
example, Bronstein~\cite{bronstein00} and Schneider~\cite{schneider04c} have observed that a
universal denominator can be constructed also when the coefficient domains are
difference fields which can be used for representing nested sums and products
($\Pi\Sigma$-fields).  For such domains, the situation is more involved.
There is a need to distinguish between ``normal'' factors of the universal
denominator which can be found very much like in the usual polynomial case,
and ``special'' factors which have to be constructed by some other means.

In the present article, we consider partial (i.e., multivariate) linear
difference equations with polynomial coefficients (PLDEs).  Our ultimate goal is
the construction of a universal denominator for potential rational solutions of
a given PLDE. Like in the univariate case with sophisticated coefficient
domains, there are two kinds of factors to be distinguished. As a matter of fact,
some parts of the denominator cannot be bound at all. For example, the equation
\[
  f(n+1,k) = f(n,k+1)
\]
has $(n+k)^{-\alpha}$ as a rational solution, for any $\alpha\in\set N$, and there
is obviously no finite polynomial $Q$ that would be a multiple of
$(n+k)^{\alpha}$ for all $\alpha\in\set N$.  We will call factors
that may exhibit such ``special'' behaviour \emph{periodic.} Our main result is
that we can construct for any given PLDE a polynomial $d$ such that every
\emph{aperiodic} factor of any potential solution $p/q$ must divide~$d$.

Such a bound on the aperiodic factors of the denominators does not directly give
rise to a full algorithm for finding rational solutions of PLDEs, but it can be
considered as a step in this direction. For a full algorithm, besides of
the bounding of the periodic parts of the denominator, also the entire question
of how to find (all) polynomial solutions of a PLDE in the third step is wide
open and far from being settled.
But even if these parts have to remain open for now, our aperiodic denominator
bound is useful in practice. When it comes to solving an actual equation, possible
periodic factors in a solution can often be guessed by inspection, their
multiplicities can be determined by trial and error, and degree bounds for
polynomial solutions can be established heuristically. A reasonably tight universal
denominator, on the other hand, cannot be as easily obtained on heuristic grounds.

\section{The Univariate Case}\label{sec:abramov}

Before entering the multivariate setting, let us summarize Abramov's classical
denominator bound for univariate equations.
We will introduce on the fly some notions and notations needed later.

Let $\set K$ be a field of characteristic zero and let $\set K[n]$ and $\set K(n)$
denote the ring of univariate polynomials and the field of rational functions in
$n$ with coefficients in~$\set K$, respectively.
Write $N$ for the shift operator acting on $\set K[n]$ and $\set K(n)$ via
\[
  Nq(n) := q(n+1).
\]
The objects of interest are difference equations of the form
\begin{equation}\label{eq:olde}
  a_0y + a_1Ny + \cdots + a_mN^my=f
\end{equation}
where $a_0,\dots,a_m,f\in\set K[n]$ ($a_0,a_m\neq0$) are given and $y\in\set K(n)$ is unknown.

The denominator bounding problem is as follows: given $a_0,\dots,a_m,f\in\set K[n]$,
find $Q\in\set K[n]\setminus\{0\}$ such that the denominator of any solution $y\in\set K(n)$
of~\eqref{eq:olde} divides~$Q$.
Abramov's denominator bounding algorithm~\cite{abramov71} is an efficient way of computing
\begin{equation}\label{Equ:AbramovDen}
  \gcd(\prod_{i=0}^s N^i a_0, \prod_{i=0}^s N^{-m-i} a_m),
\end{equation}
where
\begin{equation}\label{Equ:UnivariateDisp}
s:=\max\{\,i\geq0:\gcd(N^i a_0,N^{-m} a_m)\neq1\,\}
\end{equation}
is the
\emph{dispersion} of $a_0$ and~$N^{-m}a_m$. It is efficient in the sense that
the gcd is constructed without explicitly calculating the products.

To see that this bound is correct, write~\eqref{eq:olde} in the form
\[
  N^my=\frac1{a_m}\Bigl(f-\sum_{i<m}a_iN^iy\Bigr).
\]
Shifting this equation by $s$ gives
\[
  N^{m+s}y=\frac1{N^sa_m}\Bigl(N^sf-\sum_{i<m}(N^sa_i)(N^{i+s}y)\Bigr).
\]
By repeatedly using the recurrence, the terms $N^{i+s}y$ appearing on the
right hand side can be reduced to smaller shifts of~$y$ so that
for certain polynomials $b,b_0,\dots,b_{m-1}$ we have
\[
  N^{m+s}y=\frac{b-\sum_{i<m} b_i N^i y}{\prod_{i=0}^s N^i a_m}.
\]
At this point we rely on the following result.

\begin{theorem}\label{Thm:BoundUnivariateSol} \cite{abramov71}
For any solution $y=\frac{p}{q}\in\set K(n)$ of~\eqref{eq:olde},
$$\max\{i\geq0:\gcd(q,N^iq)\neq1\}\leq s.$$
\end{theorem}

This theorem ensures that the denominator of a solution $y$ cannot contain
two factors $u,v$ with $u=N^{s+1}v$, and this in turn implies that no denominator of
any of the $N^iy$ on the right can have a common factor with the denominator of
$N^{m+s}y$.  Therefore the denominator of $N^{m+s}y$ must be a divisor of
$\prod_{i=0}^s N^i a_m$, and therefore the denominator of $y$ must be a divisor
of
\begin{equation}\label{Equ:BoundForRightCornerPoint}
\prod_{i=0}^s N^{i-m-s} a_m=\prod_{i=0}^s N^{-m-i}a_m.
\end{equation}

By an analogous argument, now rewriting $y$ in terms of higher shifts, it can be
shown that the denominator of $y$ must divide $\prod_{i=0}^s N^i a_0$.  As both
bounds must hold simultaneously, the correctness of Abramov's bound is
established. The argument given here is not exactly Abramov's
original one, but follows a proof which to our knowledge first appeared
in~\cite{barkatou99}. The equivalence of the two approaches is shown in~\cite{chen08b}.

\section{The Multivariate Case}

For the rest of this article we will be concerned with adapting the univariate reasoning
sketched in the previous section to the multivariate
setting. From now on, we consider polynomials and rational functions in
$r$ variables $n_1,\dots,n_r$. Wherever it seems appropriate, we will use
multiindex notation, writing for instance $\vec n$ for $n_1,\dots,n_r$, etc.

We define shift operators $N_1,\dots,N_r$ acting on $\set K[\vec n]$ and
$\set K(\vec n)$ in the obvious way:
\begin{alignat*}1
  &N_i q(n_1,\dots,n_r) \\
  &\quad{}= q(n_1,\dots,n_{i-1},n_i+1,n_{i+1},\dots,n_r).
\end{alignat*}
For $\vec i=(i_1,\dots,i_r)\in\set Z^r$ we will abbreviate
\[
  N^{\vec i}q := N_1^{i_1}N_2^{i_2}\cdots N_r^{i_r}q.
\]

A partial linear difference equation (PLDE) is an equation of the form
\begin{equation}\label{Equ:PLDE}
  \sum_{\vec s\in S} a_{\vec s}N^{\vec s}y=f,
\end{equation}
where $S\subseteq\set Z^r$ is finite and nonempty (called the \emph{support}
or the \emph{shift set} or the \emph{structure set} of the equation),
$f\in\set K[\vec n]$ and $a_{\vec s}\in\set K[\vec n]\setminus\{0\}$ ($s\in S$) are
explicitly given polynomials, and $y\in\set K(\vec n)$ is an unknown rational function.
The polynomial $a_{\vec s}$ is called the coefficient of $N^{\vec s}$ (or simply
of~$\vec s$).

In the following section, we generalize the notion of dispersion to multivariate
polynomials, and we will define the notions of periodic and aperiodic polynomials.
After that, in Section~\ref{sec:cs}, we will show how to predict all the aperiodic
factors of a rational solution $y$ of~\eqref{Equ:PLDE}.






\section{Spread and Dispersion}

The notions of spread and dispersion are related to the question whether two polynomials can
be mapped to one another by a shift. In the multivariate case, we now allow independent shifts
in all directions.

Given two polynomials $p,q\in\set K[\vec n]$, we say that they are shift equivalent if
there exists $\vec i=(i_1,\dots,i_r)\in\set Z^r$ such that $p(\vec n)=q(\vec n+\vec i)$.
In operator notation, $p$~and $q$ are shift equivalent iff $p=N^{\vec i}q$.

\begin{definition}
  Let $p,q\in\set K[\vec n]$.
  The set
  \begin{alignat*}1
    &\spread(p,q):=\{\,(i_1,\dots,i_r)\in\set Z^r:\\
    &\hspace{10em}\gcd(p,N_1^{i_1}\cdots N_r^{i_r}q)\neq1\,\}
  \end{alignat*}
  is called the \emph{spread} of $p$ and~$q$.
  The number
  \[
    \disp_k(p,q):=\max\bigl\{\,|i_k| : (i_1,\dots,i_r)\in\spread(p,q)\,\bigr\}
  \]
is called the \emph{dispersion} of $p$ and~$q$ w.r.t.\ $k\in\{1,\dots,r\}$, and
  \[
    \disp(p,q):=\max(\disp_1(p,q),\dots,\disp_r(p,q))
  \]
  is called \emph{dispersion} of $p$ and~$q$. (By convention, $\max A:=-\infty$
  if $A$ is empty and $\max A:=\infty$ if $A$ is unbounded.)

  The polynomial $p$ is called \emph{periodic} if $\spread(p,p)$ is infinite and \emph{aperiodic} otherwise.
\end{definition}

Note that this definition does not exactly correspond to the definition stated before
for the univariate case. While there, the definition depends on whether one shifts
to the left or to the right~\cite{abramov71,bronstein00}, our definition takes all
directions into account. This makes the reasoning below a little simpler.

\begin{example}
  Let
  \begin{alignat*}1
    p&=\bigl(n^2+k^2\bigr)\bigl((n+1)^2+(k-3)^2\bigr)\bigl(k-n+3\bigr),\\
    q&=\bigl((n+2)^2+(k-1)^2\bigr)\bigl((n-2)^2+(k+7)^2\bigr)\bigl(2k-3n\bigr).
  \end{alignat*}
  Then, by inspection,
  \begin{alignat*}1
    \spread(p,q)&=\{(2,-1),\ (-2,7),\ (1,2),\ (-3,10)\},\\
    \disp(p,q)&=10.
  \end{alignat*}
  Both $p$ and $q$ are aperiodic. An example for a periodic polynomial is $n-k$, because,
  again by inspection,
  \begin{alignat*}1
    &\spread(n-k,n-k)=\{\,(i,i):i\in\set Z\,\}\\
    &\quad{}=\{\dots,(-1,-1),(0,0),(1,1),(2,2),\dots\}
  \end{alignat*}
  is infinite.
\end{example}

In the univariate case, the spread of two polynomials can be found as the set of all
integer roots of the polynomial $\res_n(p(n),q(n+i))\in\set K[i]$. This is no longer possible
in the multivariate setting, for in the case of several variables, common roots no longer
correspond to common factors. Nevertheless, it turns out that the multivariate spread as defined
as above can be effectively computed. Let us consider the somewhat simpler situation of
irreducible polynomials first. In this situation the spread cannot take on any cardinality:

\begin{lemma}
  Let $p,q\in\set K[\vec n]$ be irreducible and aperiodic. Then $|\spread(p,q)|\leq 1$.
\end{lemma}
\begin{proof}
  Suppose $p$ and $q$ are such that $|\spread(p,q)|>1$. Then there exist two different
  multiindices $(i_1,\dots,i_r)$ and $(j_1,\dots,j_r)$ with
  \[
    \gcd(p,N_1^{i_1}\cdots N_r^{i_r}q)\neq1
    \quad\text{and}\quad
    \gcd(p,N_1^{j_1}\cdots N_r^{j_r}q)\neq1.
  \]
  As $p$ and $q$ are irreducible and irreducibility is preserved under the shifts $n_i\mapsto n_i+1$,
  we have in fact
  \[
    c p = N_1^{i_1}\cdots N_r^{i_r}q = N_1^{j_1}\cdots N_r^{j_r}q
  \]
  for some $c\in\set K\setminus\{0\}$, hence
  \[
    q = N_1^{i_1-j_1}\cdots N_r^{i_r-j_r}q,
  \]
  in contradiction to the assumption that $q$ is aperiodic.
\end{proof}

It is not essential for the lemma that we consider only shifts of
integer distance. More generally, if $p$~and $q$ are two
irreducible polynomials, then, by the same argument, the number of vectors $\vec
i\in\bar{\set K}^r$ with $p(\vec n)=q(\vec n+\vec i)$ is either $0$ or~$1$
or infinite ($\bar{\set K}$~refers to the algebraic closure of~$\set K$).
These vectors $\vec i\in\bar{\set K}$ can be found by making a brute force
ansatz. For a variable vector $\vec i=(i_1,\dots,i_r)$, force
\[
  p(\vec n)-q(\vec n+\vec i)\stackrel!=0
\]
and compare coefficients with respect to $\vec n$ to obtain an algebraic system
of equations in $i_1,\dots,i_r$ over~$\set K$.
The solutions will form an affine linear space over $\set K$, for whenever
$\vec i,\vec j\in\bar{\set K}^r$ are such that
$p(\vec n)=q(\vec n+\vec i)$ and $p(\vec n)=q(\vec n+\vec j)$, then
\[
  p(\vec n)=q(\vec n+\vec i+ \alpha(\vec i-\vec j))
  \text{ for all $\alpha\in\set Z$},
\]
and since the solution set is Zariski-closed, what is true for all $\alpha\in\set Z$
must also be true for all $\alpha\in\set K$.

The spread of two irreducible polynomials can therefore be computed by first determining
a basis of the affine linear space of all
possible shifts $\vec i\in\bar{\set K}^r$ mapping one given polynomial to the
other. By taking the radical to remove nontrivial multiplicities, it is ensured that
a basis of the linear space can be read off from a Gr\"obner basis.
In a second step, we filter out from this affine space the vectors
which have integral coordinates only:

\begin{algorithm}\label{Alg:Spread}
Input: $p,q\in\set K[\vec n]$ irreducible\\
Output: $\spread(p,q)$

\begin{algo}%
S := 'Coefficients'(p(\vec n)-q(\vec n+\vec i), \{\vec n\})\subseteq\set K[\vec i]
G := 'Gr\ddot obnerBasis'('Radical'(S), 'degrevlex'(\vec i))
\vtop{\hsize=.9\hsize\hangindent=2em\hangafter=1\noindent %
 Choose a basis~$B$ of the $\set Q$ vector space generated by the coefficients of %
 the elements of $G$, say $B=\{b_1,\dots,b_d\}\subseteq\set K$.}
\vtop{\hsize=.9\hsize\hangindent=2em\hangafter=1\noindent %
 For each $g\in G$, let $g^{(1)},\dots,g^{(d)}\in\set Q[\vec i]$ be such that %
 $g=b_1g^{(1)}+b_2g^{(2)}+\cdots+b_dg^{(d)}$. (Note: At this point all $g\in G$ %
 are linear, and so are all the $g^{(k)}$.)}
S := \bigcup_{g\in G}\bigl\{g^{(1)},g^{(2)},\dots,g^{(d)}\bigr\}
|return| \ker_{\set Z}(S)%
 \end{algo}
\end{algorithm}

Note that the algorithm avoids the need of solving systems of diophantine
equations by exploiting a priori knowledge on the structure of the solution set.

The case of non-irreducible polynomials is easily reduced to the former case
by considering all pairs of factors. To be precise, let $p,q\in\set K[\vec n]\setminus\{0\}$ be
any polynomials, and let
\[
  p=p_1^{u_1}p_2^{u_2}\cdots p_r^{u_r},\qquad
  q=q_1^{v_1}q_2^{v_2}\cdots q_s^{v_s}
\]
be their factorization into irreducible factors. Then
\[
 \spread(p,q)=\bigcup_{i=1}^r\bigcup_{j=1}^s\spread(p_i,q_j).
\]
In short, given $p,q\in\set K[\vec n]\setminus\{0\}$ we can compute
$\spread(p,q)$ and therefore also $\disp_i(p,q)$ and~$\disp(p,q)$.

Every given polynomial $p$ can be split uniquely (up to constant multiples)
into a factorization $p=uv$ where $u$ is periodic and $v$ is aperiodic.
We call $u$ and $v$ the \emph{periodic} and \emph{aperiodic part} of~$p$, respectively.
As we can factor polynomials and compute, as described above, their spread,
this decomposition can be computed.

\begin{example}
\begin{enumerate}
\item $p=n+k$, $q=n+2k$. Here we have
  \[
    p(n,k)-q(n+i,k+j)=(-i+2j)-k
  \]
  and $G=\{1\}$, hence $\spread(p,q)=\emptyset$.
\item $p=n+\sqrt2k$, $q=n+\sqrt2k+3-2\sqrt2$. Here we have
  \[
    p(n,k)-q(n+i,k+j)=- i - \sqrt2 j-3 + 2 \sqrt2
  \]
  and $G=\{- i - \sqrt2 j-3 + 2 \sqrt2 \}$.
  With $B=\{1,\sqrt2\}$ we get
  $S=\{-i-3,-j+2\}$,
  from which we obtain
  \[
   \spread(p,q)=\{(-3,2)\}.
  \]
\item $p=3 k^2+6 k n-7 k+3 n^2-7 n+1$, $q=3 k^2+6 k n-13 k+3 n^2-13 n+11$.
  Here we have
  \begin{alignat*}1
    &p(n,k)-q(n+i,k+j)=-(i + j-1) (3 i + 3 j -10)\\
    &\qquad{}-6 (i+j-1)k+6 (i+j-1)n
  \end{alignat*}
  and $G=\{i+j-1\}$. With $B=\{1\}$, we get $S=G$ from which we obtain
  \[
   \spread(p,q)=\begin{pmatrix}1\\0\end{pmatrix}+\set Z\begin{pmatrix}1\\-1\end{pmatrix}.
  \]
\end{enumerate}
\end{example}

\section{Aperiodic Factors in Denominators of Solutions}\label{sec:cs}

In this section we solve the following problem. \textbf{Given} a nonempty finite
shift set $S\subseteq\set Z^r$ with coefficients $a_{\vec s}\in\set K[\vec
n]\setminus\{0\}$ ($\vec s \in S$), \textbf{find} a polynomial $d\in\set K[\vec
n]\setminus\{0\}$ such that for any solution $y=\frac{p}{u\,q}$
of~\eqref{Equ:PLDE} with $p,q,u\in\set K[\vec n]$ where $q$ is aperiodic and $u$
is periodic, we have $$q\mid d.$$ Such a $d$ is called an \emph{aperiodic
  universal denominator} (or \emph{aperiodic denominator bound})
of~\eqref{Equ:PLDE}.

First we generalize Theorem~\ref{Thm:BoundUnivariateSol}, i.e., for any solution $y\in\set K(\vec n)$ of~\eqref{Equ:PLDE} we bound the dispersion of the aperiodic denominator part of $y$. To be more precise, we first bound the dispersion w.r.t.\ one component $n_i$ in $\vec{n}$.

\begin{lemma}\label{Lemma:DispBound}
Let $S\subseteq\set Z^r$ be finite and nonempty and let $a_{\vec s}\in\set K[\vec n]\setminus\{0\}$ for $\vec s \in S$, and $f\in\set K[\vec n]$; let $a'_{\vec{s}}$ be the aperiodic part of $a_{\vec{s}}$. Let $i\in\{1,\dots,r\}$. Define
$$k:=\max\{|a_i-b_i|: (a_1,\dots,a_r),(b_1\dots,b_r)\in S\}$$
and let
\begin{align*}
A=&\{(s_1,\dots,s_r):(s_1,\dots,s_r)\in S
\text{ and }\\
&\quad\exists (t_1,\dots,t_r)\in S\text{ s.t.\ } t_i-s_i=k\},\\
B=&\{(s_1,\dots,s_r):(s_1,\dots,s_r)\in S\text{ and }\\
&\quad\exists (t_1,\dots,t_r)\in S\text{ s.t.\ } s_i-t_i=k\}.
\end{align*}
Define
$$s_i:=\max\{\disp_i(a'_{\vec{s}},N^{-k}_{i}a'_{\vec{t}}):\;\vec{s}\in A\text{ and } \vec{t}\in B\}.$$
Then for any solution $y=\frac{p}{u\,q}\in\set K(\vec n)$ of~\eqref{Equ:PLDE} with periodic part $u$ and aperiodic part $q$ we have
$$\disp_i(q)=\disp_i(q,q)\leq s_i.$$
\end{lemma}

\smallskip
\centerline{\epsfig{file=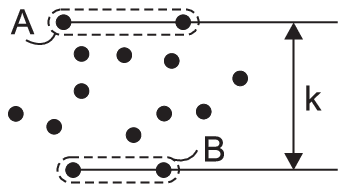}}
\smallskip

\begin{proof}
Suppose that $d:=\disp_i(q)>s_i$. Then we find irreducible factors $u$ and $v$ of $q$ such that
\begin{equation}\label{Equ:DistanceTooBig}
N^{\vec e}u=v
\end{equation}
for some $\vec e=(e_1,\dots,e_r)\in\set Z^r$ with $e_i=d$. Consider all the factors $N^{\vec u}u$ and $N^{\vec v}v$ occurring  in $q$ where the $i$th entries in $\vec u$ and $\vec v$ are $0$.
We choose now those factors from $q$ where $\vec u$ and $\vec v$ are maximal w.r.t.\ lexicographic order; these factors are denoted by $u'$ and $v'$, respectively.\\  
First suppose that $u'$ divides one of the polynomials $a_{\vec{s}}$ with $\vec{s}\in A$. As $S$ is not empty, $B$~is not empty.
Therefore, we can choose that polynomial $a_{\vec w}$ with $\vec{w}=(w_1,\dots,w_r)\in B$ such that $(w_1,\dots,w_{i-1},w_{i+1},\dots,w_r)$ is maximal w.r.t.\ lexicographic order. By~\eqref{Equ:PLDE} we can write
\begin{equation}\label{Equ:RewritePLDE}
N^{\vec w}y=\frac{1}{a_{\vec{w}}}\Big(f-\sum_{\vec{s}\in S\setminus\{\vec{w}\}}a_{\vec{s}}N^{\vec{s}}y\Big).
\end{equation}
Now observe that the factor $N^{\vec w}v'$ does not occur in the denominator of any $N^{\vec s}y$ with $\vec s\in S\setminus\{\vec{w}\}$: if $N^{\vec w}v'$ occurred in such a denominator, $\vec{s}\in B$ by construction (recall that $u'$ and $v'$ have maximal distance $d$ in the $i$-coordinate among all factors in $q$ and that $v'$ with $N^{\vec w}v'$ is shifted maximally by $k$ among all possible choices from $S$ in direction of the $i$-coordinate since $\vec{w}\in B$; so only if $\vec{s}\in B$ is necessary to guarantee that  $N^{\vec w}v'$ is a factor of $N^{\vec s}q$ and hence of the denominator of $N^{\vec s}y$). But
then by the assumption that $(w_1,\dots,w_{i-1},w_{i+1},\dots,w_r)$ is maximal w.r.t.\ lexicographic order (among all possible choices) and that $v'=N^{\vec v}v$ with $\vec v=(v_1,\dots,v_{i-1},0,v_{i+1},\dots,v_r)$ is maximal w.r.t.\ lexicographical order, it follows that $N^{\vec w}v'$ can only occur in the denominator of $N^{\vec w}y$. Summarizing, the factor $N^{\vec w}v'$ does not occur in the denominators of $N^{\vec{s}}y$ for any $s\in S\setminus\{\vec{w}\}$, and since $f,a_{\vec{s}}\in\set K[\vec{n}]$, the common denominator of
$f-\sum_{\vec{s}\in S\setminus\{\vec{w}\}}a_{\vec{s}}N^{\vec{s}}y$ does not contain the factor  $N^{\vec w}v'$. Moreover, since $u'$ is a factor of $a_{\vec{s}}$ for some $\vec{s}\in A$, and since $\vec{w}\in B$ and $d>s_i$, $N^{\vec w}v'$ cannot be a factor of $a_{\vec{s}}$ for any $\vec{s}\in B$. In particular, our $a_{\vec{w}}$ is not divisible by $N^{\vec w}v'$. Overall, the common denominator of the right hand side of~\eqref{Equ:RewritePLDE} cannot contain the factor $N^{\vec w}v'$ which implies that the denominator of $N^{\vec w}y$ is not divisible by $N^{\vec w}v'$. Thus the denominator of $y$, in particular $q$ is not divisible by $v'$; a contradiction.

Conversely, suppose that $u'$ does not divide any of the polynomials $a_\vec{s}$ with $\vec{s}\in A$. Then by similar arguments as above (the roles of $A$ and $B$ exchanged), we derive again a contradiction. Therefore $s_i\leq\disp_i(q)$.
\end{proof}

\begin{example}\label{Exp:DispConnection}
In the generic univariate case~\eqref{eq:olde} ($r=i=1$) the shift set is $S=\{0,1,\dots,m\}\subseteq\set Z^1$ and for the sets $A$ and $B$ from Lemma~\ref{Lemma:DispBound} we have
$A=\{0\}$ and $B=\{m\}$.
In this particular instance, Lemma~\ref{Lemma:DispBound} corresponds to Theorem~\ref {Thm:BoundUnivariateSol}.
\end{example}

A bound of the dispersion for the multivariate case is given in the next theorem.

\begin{theorem}\label{Thm:DispBound}
Let $S\subseteq\set Z^r$ be finite and nonempty and let $a_{\vec s}\in\set K[\vec n]\setminus\{0\}$ for $\vec s \in S$, and $f\in\set K[\vec n]$. Then one can compute an $s\in\set N\cup\{-\infty\}$ with the following property: For any solution $y=\frac{p}{u\,q}\in\set K(\vec n)$ of~\eqref{Equ:PLDE} with periodic part $u$ and aperiodic part $q$ we have
\begin{equation}\label{Equ:DisBound}
\disp(q)=\disp(q,q)\leq s.
\end{equation}
\end{theorem}
\begin{proof}
Compute the values $s_i$ for $i\in\{1,\dots,r\}$ as described in Lemma~\ref{Lemma:DispBound}; the spread can be computed with Algorithm~\ref{Alg:Spread}. By taking
$s=\max(s_1,\dots,s_r)$ the property $\disp(q)=\disp(q,q)\leq s$ is guaranteed.
\end{proof}

In order to derive an aperiodic denominator bound, we adapt the idea presented in Section~\ref{sec:abramov}. Namely, we will choose an appropriate point $\vec{p}\in S$ and express $N^\vec{p}y$ for any solution $y\in\set K(\vec{n})$ of~\eqref{Equ:PLDE} in terms of $N^\vec{s}y$ for points $\vec{s}\in S'$ which are sufficiently far away from $\vec{p}$. To be more precise, for any $s>0$ we can explicitly write
\begin{equation}\label{Equ:ShiftRepresenation}
N^{\vec{p}}y=\frac{b+\sum_{\vec{i}\in S'}b_{\vec{i}}N^{\vec i}y}{\prod_{\vec{i}\in W-\vec{p}}N^{\vec{i}}a_{\vec{p}}}
\end{equation}
for some polynomials $b,b_{\vec{i}}\in\set K[\vec{n}]$ and for finite sets $W,S'\subseteq\set Z^r$ with the following property: the distance of the points $S'$ to $\vec{p}$ is at least $s$. Then by taking $s$ as in Theorem~\ref{Thm:DispBound}, we will be able to conclude that $\prod_{\vec{i}\in W-\vec{p}}N^{\vec{i}}a_{\vec{p}}$ is an aperiodic denominator bound of $N^{\vec{p}}y$, and consequently $\prod_{\vec{i}\in W}N^{\vec{i}-2\vec{p}}a_{\vec{p}}$ is an aperiodic denominator bound of $y$.

Such appropriate points $\vec{p}$ from $S$ can be chosen as follows.
Let $S\subseteq\set Z^r$ be a finite set. A point $\vec{p}\in S$ is called a \emph{corner point} (or \emph{extreme point}) of $S$, if there exists an affine hyperplane $H$ (codimension $1$) which contains $\vec{p}$ and where all other points $S\setminus\{\vec{p}\}$ are situated in one of the two open half spaces determined by the hyperplane.

\smallskip
\centerline{\epsfig{file=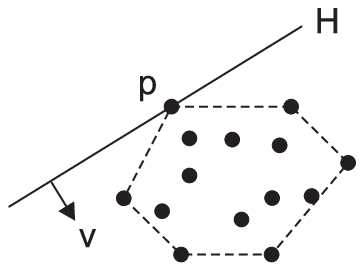}}
\smallskip

\noindent Such an affine hyperplane $H$ is called \emph{border plane} of $S$ for $\vec{p}$, and a vector being orthogonal to $H$ and directing to the half space of the points $S\setminus\{\vec{p}\}$ is called \emph{inner vector}.

Note that the corner points are the extreme points of the convex hull generated by $S$, and they can be computed by simple linear algebra; for further details see, e.g., \cite{Lang87}.

\begin{example}
In the generic univariate case~\eqref{eq:olde} ($r=1$ and $S=\{0,1,\dots,m\}$) the corner points are $0$ and $m$, and the border planes are $\{0\}$ and $\{m\}$ with inner vectors $1$ and $-1$, respectively.

More generally, if we are given a finite set $S\subseteq\set Z^r$ with $(0,\dots,0)\in S$ and $\max\{d_i:(d_1,\dots,d_r)\in S\}>0$ for each $1\leq i\leq r$, there are at least $r+1$ corner points.
\end{example}

For our denominator bound construction we start with the following simple lemma.

\begin{lemma}\label{Lemma:ShiftSystem}
Let $S\subseteq\set Z^r$ be a nonempty finite set and let $\vec p\in S$ be a corner point together with a border plane $H$ and inner vector $\vec{v}$. Consider any hyperplane $H'$ which is parallel to $H$. Then for any $\vec{p'}\in H'\cap\set Z^r$ the points $S+(\vec{p'}-\vec{p})\setminus\{\vec{p'}\}$ are all outside of $H'$ in the half space determined by the direction of $\vec{v}$.
\end{lemma}

\smallskip
\centerline{\epsfig{file=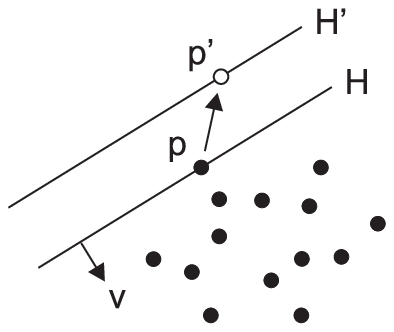}}
\smallskip

\begin{proof}
Since $H+(\vec{p'}-\vec{p})=H'$, $H'$ is a border plane of $S+(\vec{p'}-\vec{p})$ for $\vec{p'}$. This proves the lemma.
\end{proof}

By iterative application of the previous lemma we obtain the following theorem.

\begin{theorem}\label{Thm:ShiftRecurrence}
Let $S\subseteq\set Z^r$ be nonempty and finite, and let $a_{\vec s}\in\set K[\vec n]\setminus\{0\}$ for $\vec s \in S$, and $f\in\set K[\vec n]$. Let $\vec p$ be a corner point of $S$ with a border plane $H$ and inner vector $\vec v$. Then for every $s>0$ there exist finite sets
\begin{align}
W&\subseteq\set Z^r\cap \bigcup_{0\leq e\leq s}(H+e\vec{v})\text{ and}\label{Equ:W}\\
S'&\subseteq\set Z^r\cap \bigcup_{e>0}(H+(s+e)\vec{v}),\label{Equ:S'}
\end{align}
and polynomials $b,b_{\vec{i}}\in\set K[\vec n]$ such that for any solution $y\in\set K(\vec n)$ of~\eqref{Equ:PLDE} the relation~\eqref{Equ:ShiftRepresenation} holds.
\end{theorem}

\smallskip
\centerline{\epsfig{file=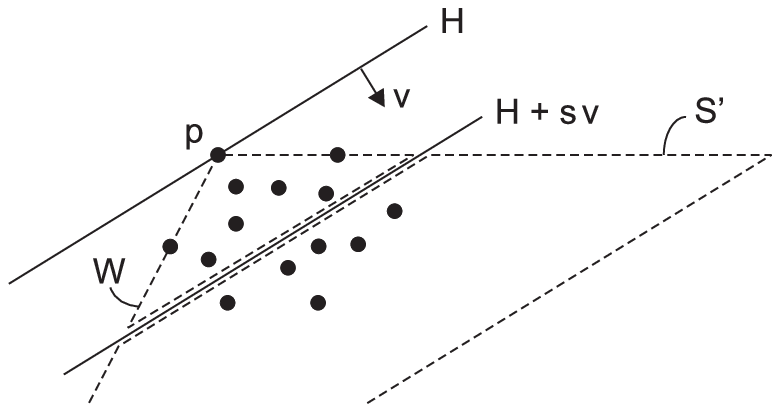}}
\smallskip

\begin{proof}
We show the theorem for a generic solution $y\in\set K(\vec{n})$ of~\eqref{Equ:PLDE}. Hence the ingredients  $b,b_{\vec{i}}\in\set K[\vec n]$ and $W,S'$ under consideration will hold for any specific solution as stated in the theorem.
For $\tilde{S}:=S\setminus\{\vec p\}$ we have
\begin{equation*}
N^{\vec{p}}y=\frac{f-\sum_{\vec s\in\tilde{S}}a_{\vec s}N^{\vec s}y}{a_{\vec{p}}}
\end{equation*}
by~\eqref{Equ:PLDE}. If $\tilde{S}=\{\}$, take $S'=\{\}$ and $W=\{\vec{p}\}$ and we are done. Otherwise, let $\vec{p'}\in\tilde{S}$ be such that the distance to $H$ is minimal. Define $H':=H+(\vec{p'}-\vec{p})$. If $H'\subseteq \{H+(s+e)\vec{v}:e>0\}$, we are again done with $W=\{\vec{p}\}$ and $S'=S\setminus\{\vec p\}$.\\
Now suppose that $H'\subseteq\{H+e\vec{v}:0\leq e\leq s\}$, and let $\{\vec{p}_1,\dots,\vec{p}_k\}=H'\cap \tilde{S}$ (by construction $\vec{p'}\in H'\cap \tilde{S}$) and define $S_1=\tilde{S}\setminus\{\vec{p}_1,\dots,\vec{p}_k\}$. Then we can write
\begin{equation}\label{Equ:PLDEOrg}
N^{\vec{p}}y=\frac{f-\sum_{\vec s\in S_1}a_{\vec s}N^{\vec s}y-\sum_{\vec s\in\{\vec{p}_1,\dots,\vec{p}_k\}}a_{\vec s}N^{\vec s}y}{a_{\vec{p}}}.
\end{equation}
By~\eqref{Equ:PLDE}
\begin{equation}\label{Equ:PLDEShift}
N^{\vec{p}_1}y=\frac{N^{\vec{p}_1-\vec{p}}f-\sum_{\vec s\in \tilde{S}+(\vec{p}_1-\vec{p})}a_{\vec s}N^{\vec s}y}{N^{\vec{p}_1-\vec{p}}a_{\vec{p}}}.
\end{equation}
In particular, by Lemma~\ref{Lemma:ShiftSystem} the points from
$\tilde{S}+(\vec{p}_1-\vec{p})$ are all outside of $H'$ in the half space determined by the direction of $\vec{v}$. Thus we   substitute~\eqref{Equ:PLDEShift} into~\eqref{Equ:PLDEOrg} and can express $N^{\vec{p}}y$ in the form
$$N^{\vec{p}}y=\frac{f'-\sum_{\vec s\in S'_1}a'_{\vec s}N^{\vec s}y-\sum_{\vec s\in\{\vec{p}_2,\dots,\vec{p}_k\}}a_{\vec s}N^{\vec s}y}{a_{\vec{p}}N^{\vec{p}_1-\vec{p}}a_{\vec{p}}}
$$
for some $f'\in\set K[\vec n]$ and $a'_{\vec s}\in\set K[\vec n]$ with
$S'_1=S_1\cup(\tilde{S}+\vec{p}_1-\vec{p})$. After $k-1$ further reductions, we
end up at the form
$$N^{\vec{p}}y=\frac{b-\sum_{\vec s\in S_1''}b_{\vec s}N^{\vec s}y}{\prod_{\vec{i}\in W_1-\vec{p}}N^\vec{i}a_{\vec{p}}}$$
for some $b,b_{\vec{s}}\in\set K[\vec{n}]$ and with
$W_1=\{\vec{p},\vec{p}_1,\dots,\vec{p}_k\}$ and $S_1''=S_1\cup \bigcup_{1\leq
  i\leq k}\tilde{S}+(\vec{p_i}-\vec{p})$.  Note that all points
$\vec{p}_1,\dots,\vec{p}_k$ which are closest to $H$ have been eliminated.  Now
we repeat the construction from above until we enter in the base case given in
the beginning of the proof. This completes the proof.
\end{proof}

Note that all the ingredients $W$, $S'$, $b\in\set K[\vec n]$ and the $b_{\vec{s}}$ for $\vec{s}\in W$ can be computed explicitly. However, for getting an aperiodic denominator bound, we only need $W$. The proof of Theorem~\ref{Thm:ShiftRecurrence} delivers the following simple algorithm.

\begin{algorithm}\label{Alg:FindStructureSet}
Input: A finite nonempty set $S\subseteq\set Z^r$ and a corner point $\vec{p}$ of $S$ with a border plane $H$ and inner vector $\vec v$; $s>0$.\\
Output: A finite set $W\subseteq\set Z^r$ with~\eqref{Equ:W} such that there are $S'$ with~\eqref{Equ:S'} and $b,b_{\vec{i}}\in\set K[\vec n]$ such that \eqref{Equ:ShiftRepresenation} holds for any solution $y\in\set K(\vec{n})$ of~\eqref{Equ:PLDE}.

\begin{algo}%
\tilde{S}:=S\setminus\{\vec p\}; {S}^\prime:=S\setminus\{\vec p\}; W:=\{\vec{p}\}
|while| \bigcup_{0\leq e\leq s}(H+e\vec{v})\cap S^\prime\neq \{\} |do|
\vtop{\hsize=.9\hsize\hangindent=2em\hangafter=1\noindent%
^^ILet $\{\vec{p_1},\dots,\vec{p_k}\}$ be the points in $S^\prime$ which have minimal distance to $H$.%
\vadjust{\kern2.5pt}}
^^I W:=W\cup\{\vec{p_1},\dots,\vec{p_k}\}
^^I S^\prime:=S^\prime\setminus\{\vec{p_1},\dots,\vec{p_k}\}
^^I S^\prime:=S^\prime\cup \bigcup_{1\leq i\leq k}\tilde{S}+(\vec{p_i}-\vec{p})
|enddo|
|return|W%
 \end{algo}
\end{algorithm}

Finally, we end up at the following theorem which tells us how we can compute an aperiodic denominator bound.

\begin{theorem}\label{Thm:UniversalBound}
Let $S\subseteq\set Z^r$ be finite and nonempty, and let
$a_{\vec s}\in\set K[\vec n]\setminus\{0\}$ for $\vec s \in S$, and $f\in\set K[\vec n]$; let $a'_{\vec{s}}$ be the aperiodic part of $a_{\vec{s}}$.\\
Take $s\in\set N\cup\{-\infty\}$ s.t.\ for any solution $y=\frac{p}{u\,q}\in\set K(\vec n)$ of~\eqref{Equ:PLDE} with periodic part $u$ and aperiodic part $q$ we have~\eqref{Equ:DisBound}.\\
Let $\vec p$ be a corner point of $S$ with a border plane $H$ and inner vector $\vec v$ with $|\vec v|\geq1$. Let $W$ be the output of Algorithm~\ref{Alg:FindStructureSet} with input $H$, $\vec p$ and $s$.
Then
$$\prod_{\vec{s}\in W-2\vec{p}}N^{\vec{s}}a'_{\vec{p}}$$
is an aperiodic universal denominator of~\eqref{Equ:PLDE}.
\end{theorem}
\begin{proof}
Let $y=\frac{p}{u\,q}\in\set K(\vec n)$ be a solution  of~\eqref{Equ:PLDE} with periodic part $u$ and aperiodic part $q$.
By construction of $W$ it follows that there are $S'$ and $b,b_{\vec{i}}\in\set K[\vec n]$ with~\eqref{Equ:S'} and~\eqref{Equ:ShiftRepresenation}. Since $|\vec{v}|\geq1$, the distance of the points $S'$ to $\vec{p}$ is larger than $s$.
Observe that $N^{\vec{p}}q$ must occur in the denominator of the right hand side of~\eqref{Equ:ShiftRepresenation}. Using~\eqref{Equ:DisBound}, the aperiodic denominator parts in $N^\vec{s}q$ with $\vec{s}\in S'$ cannot contribute to the aperiodic denominator part of $N^{\vec{p}}q$. Hence only the polynomial $\prod_{\vec{s}\in W-\vec{p}}N^{\vec{s}}a_{\vec{p}}$ is responsible for $N^\vec{p}q$, i.e., $$N^\vec{p}q\mid\prod_{\vec{s}\in W-\vec{p}}N^{\vec{s}}a_{\vec{p}}.$$ Thus $\prod_{\vec{s}\in W-2\vec{p}}N^{\vec{s}}a_{\vec{p}}$ is a universal denominator of~\eqref{Equ:PLDE}.
\end{proof}

\begin{example}
In the generic univariate case~\eqref{eq:olde} ($r=1$ and $S=\{0,\dots,m\}$) the coefficients are given by
$a_i$ for $0\leq i\leq m$. First, we take the corner point $m$ with the border plane $H=\{m\}$ and inner vector $-1$. Note that with~\eqref{Equ:UnivariateDisp} it follows that~\eqref{Equ:DisBound} holds
for any solution $y=\frac{p}{q}\in\set K(\vec n)$; see also Example~\ref{Exp:DispConnection}. Applying Algorithm~\ref{Alg:FindStructureSet} with the input $S$, $m$, $H$, $-1$, we get
$W=\{m,m-1,\dots,m-s\}$. Hence Theorem~\ref{Thm:UniversalBound} delivers the universal denominator bound (here we have only aperiodic factors)
$$\prod_{\vec{s}\in W-2m}N^{\vec{s}}a_{m}$$
which agrees with~\eqref{Equ:BoundForRightCornerPoint}. Similarly, if we take the corner point $0$ with the inner vector $1$, we get $W=\{0,1,\dots,s\}$ and obtain the universal denominator bound
$$\prod_{\vec{s}\in W-2\cdot 0}N^{\vec{s}}a_{0}=\prod_{i=0}^sN^ia_0.$$
Combining these two estimates produces~\eqref{Equ:AbramovDen}.
\end{example}

{}From the point of view of application the following remarks are in place.

\begin{enumerate}
\item $s\in \set N\cup\{-\infty\}$ can be computed by Theorem~\ref{Thm:DispBound} and by applying Algorithm~\ref{Alg:Spread}.

\item Applying Algorithm~\ref{Alg:FindStructureSet} we can compute the finite set $W\subseteq\set Z^r$; here we remark that different choices of the border plane $H$ might lead to sets $W$ of different size. Exploiting the particular structure of $S$ gives room for improvement.

\item Suppose that we are given $k$ corner points with corresponding border planes. Then by Theorem~\ref{Thm:UniversalBound} we end up at different aperiodic universal denominators, say $d_1,\dots,d_k\in\set K[\vec{n}]$.
    Then taking
    $$\gcd(d_1,\dots,d_k)$$
    leads to a sharper universal bound.
\item The coefficients $a_{\vec{s}}$ with $\vec{s}\in S$ are often available in factorized form. Then also the $d_i$s are obtained in factorized form, and the gcd-computations boil down to comparisons of these factors and bookkeeping of their multiplicities.
\end{enumerate}

Combining the aperiodic denominator bounds for different corner points gives the following result.

\begin{theorem}\label{Thm:MustContainFactors}
Let $S\subseteq\set Z^r$ be finite and nonempty and let $a_{\vec s}\in\set K[\vec n]\setminus\{0\}$ for $\vec s \in S$, and $f\in\set K[\vec n]$. Let $\vec{p}_1,\dots,\vec{p}_k\in S$ be corner points of $S$. If the denominator of a rational solution of~\eqref{Equ:PLDE} contains an aperiodic irreducible factor, then shift equivalent factors occur in each of the coefficients $a_{\vec{p}_1},\dots,a_{\vec{p}_k}$.
\end{theorem}
\begin{proof}
By Theorem~\ref{Thm:UniversalBound} aperiodic denominator bounds can be derived by the corner points $\vec{p}_j$ in the form $\prod_{\vec{i}}N^{\vec{i}}a_{\vec{p}_j}$, respectively. Hence an aperiodic denominator bound of~\eqref{Equ:PLDE} can be written in the form
$$d=\gcd(\prod_{\vec{i}}N^{\vec{i}}a_{\vec{p}_1},\dots,\prod_{\vec{i}}N^{\vec{i}}a_{\vec{p}_k}).$$
If a rational solution contains an aperiodic irreducible factor $h$, then $h$ is also contained in $d$. Hence $h$ or a shift equivalent factor occurs in each of the $a_{\vec{p}_1},\dots,a_{\vec{p}_k}$.
\end{proof}

The following special cases are immediate.

\begin{corollary}
Let $S\subseteq\set Z^r$ be finite and nonempty, and let $a_{\vec s}\in\set K[\vec n]\setminus\{0\}$ for $\vec s \in S$, and $f\in\set K[\vec n]$. Let $\vec{s},\vec{t}\in S$ be two corner points and let $a'_{\vec{s}}$ and $a'_{\vec{t}}$ be the aperiodic parts of the coefficients $a_{\vec{s}}$ and $a_{\vec{t}}$, respectively. If $\disp(a'_{\vec{s}},a'_{\vec{t}})=-\infty$, then the aperiodic denominator part of any rational solution of~\eqref{Equ:PLDE} is $1$.
\end{corollary}

\begin{corollary}
Let $S\subseteq\set Z^r$ be finite and nonempty, and let $a_{\vec s}\in\set K[\vec n]\setminus\{0\}$ for $\vec s \in S$, and $f\in\set K[\vec n]$. If there is a corner point of $S$ whose coefficient has no aperiodic factor, then the aperiodic denominator part of any rational solution of~\eqref{Equ:PLDE} is $1$.
\end{corollary}

Besides these structural consequences, Theorem~\ref{Thm:MustContainFactors} provides the following improvement of our aperiodic denominator bound algorithm. To be more precise, Lemma~\ref{Lemma:DispBound} and thus Theorem~\ref{Thm:DispBound} can be improved in the following way. In the proof of Lemma~\ref{Lemma:DispBound} we assume that there are irreducible factors $u$ and $v$ in the denominator of the solution $y\in\set K(\vec{n})$ of~\eqref{Equ:PLDE} such that~\eqref{Equ:DistanceTooBig}
for some $\vec e=(e_1,\dots e_r)\in\set Z^r$ with $e_i=d$ where $d$ is larger than $s_i$. By the choice of $s_i$ this leads to a contradiction. Now we exploit in addition Theorem~\ref{Thm:MustContainFactors}: the factors $u$ and $v$ can be only factors that occur --up to shift equivalence-- in each coefficient of the corner points $\vec{p}_1,\dots,\vec{p}_k$. Hence it suffices to choose $s_i$
as summarized in the following proposition.

\begin{proposition}
Let $S\subseteq\set Z^r$ be finite and nonempty, and let $a_{\vec s}\in\set K[\vec n]\setminus\{0\}$ for $\vec s \in S$, and $f\in\set K[\vec n]$. Let $\vec{p}_1,\dots,\vec{p}_k$ be corner points of $S$, and let
$a'_{\vec{s}}$ be the aperiodic part of $a_{\vec{s}}$ whose factors are present --up to shift equivalence-- in each coefficient of the corner points. Let $i\in\{1,\dots,r\}$. Define
$$k:=\max\{|a_i-b_i|: (a_1,\dots,a_r),(b_1\dots,b_r)\in S\}$$
and let
\begin{align*}
A=&\{(s_1,\dots,s_r):(s_1,\dots,s_r)\in S
\text{ and }\\
&\quad\exists (t_1,\dots,t_r)\in S\text{ s.t.\ } t_i-s_i=k\}\\
B=&\{(s_1,\dots,s_r):(s_1,\dots,s_r)\in S\text{ and }\\
&\quad\exists (t_1,\dots,t_r)\in S\text{ s.t.\ } s_i-t_i=k\}.
\end{align*}
Define
$$s_i:=\max\{\disp_i(a'_{\vec{s}},N^{-k}_{i}a'_{\vec{t}}):\;\vec{s}\in A\text{ and } \vec{t}\in B\}.$$
Then for any solution $y=\frac{p}{u\,q}\in\set K(\vec n)$ of~\eqref{Equ:PLDE} with periodic part $u$ and aperiodic part $q$ we have
$$\disp_i(q)=\disp_i(q,q)\leq s_i.$$
\end{proposition}

\section{Examples}

\begin{example}
Consider the recurrence
\begin{alignat*}1
 &(2 k n+1) (6 k^2+12 k-4 n^2-4 n+5) f(n,k)\\
 &\quad{}+(2 k n+4 k+1) (6 k^2+10 k+4 n^2+8 n-7) f(n+1,k)\\
 &\quad{}-(2 k n+8 n+1) (6 k^2+24 k+4 n^2-20 n-7) f(n,k+2)\\
 &\quad{}-(2 k n+4 k+8 n+17)\\
 &\qquad\quad{}\times(6 k^2+22 k-4 n^2+16 n+45) f(n+1,k+2)=0.
\end{alignat*}
The maximum spread among the coefficients of this recurrence is~$s=4$.

Every point in the shift set $\{(0,0),\penalty0 (1,0),\penalty0 (0,2),\penalty0 (1,2)\}$
qualifies as a corner point. We choose $\vec p = (0,0)$ as corner point
and let $H$ be the plane through $\vec p$ orthogonal to $\vec v=(1,1)$.

Algorithm~\ref{Alg:FindStructureSet} delivers the set
\begin{alignat*}1
    \{&(0,0),(1,0),(2,0),(3,0),(4,0),(5,0),(6,0),(7,0),(8,0),\\
      &(0,2),(1,2),(2,2),(3,2),(4,2),(5,2),(6,2)\\
      &(0,4),(1,4),(2,4),(3,4),(4,4),\\
      &(0,6),(1,6),(2,6),\\
      &(0,8)\}
\end{alignat*}
as $W$, from which by Theorem~\ref{Thm:UniversalBound} it follows that
\[
  \prod_{i=0}^8\prod_{j=0}^{4-\lceil i/2\rceil} N^i K^{2j} (2 k n+1) (6 k^2+12 k-4 n^2-4 n+5)
\]
is a universal aperiodic denominator.

Taking instead $(1,2)$ as corner point gives the aperiodic denominator bound
\begin{alignat*}1
  &\prod_{i=-8}^0\prod_{j=-4-\lfloor i/2\rfloor}^0 N^{i-1} K^{2j-2}\bigl((2 k n+4 k+8 n+17)\\[-5pt]
  &\qquad\qquad\qquad\qquad{}\times(6 k^2+22 k-4 n^2+16 n+45)\bigr).
\end{alignat*}
The greatest common divisor of the two polynomials is
\[
  (2 k n+1) (2 (k+2) n+1) (2 k (n+1)+1) (2 (k+2) (n+1)+1).
\]
This is exactly the denominator of the actual solution
\[
  \tfrac{3 k+n}{(2 k n+1) (2 (k+2) n+1) (2 k (n+1)+1) (2 (k+2) (n+1)+1)}
\]
of the recurrence.

The computation could have been simplified by disregarding the factors
$p=6 k^2+12 k-4 n^2-4 n+5$ and $q=6 k^2+22 k-4 n^2+16 n+45$. Because of
\[
  \spread(p,q)=\{\},
\]
they cannot contribute to the universal denominator (compare Theorem~\ref{Thm:MustContainFactors}).
\end{example}

\begin{example}
Some corner points may be easier to handle than others.
As an example, consider the equation
\begin{alignat*}1
  &(k^2+n^2+1) (2 k^4+4 k^3+4 k^2 n^2+4 k^2 n+6 k^2+4 k n^2\\
  &\quad{}+8 k n+9 k+2 n^4+4 n^3+6 n^2+3 n+4) f(n,k)\\
  &-(k^2+4k+n^2+5) (2 k^4+4 k^3+4 k^2 n^2+4 k^2 n+6 k^2-1\\
  &\quad{}+4 k n^2-2 k n-k+2 n^4+4 n^3+6 n^2-2 n)f(n,k+1)\\
  &-(2 k+1) (n+1) (k^2+n^2+4 n+5) f(n+1,k)=0.
\end{alignat*}
Without any computation, it can be deduced that any potential aperiodic
factor in a denominator must be a shifted copy of $k^2+n^2+4 n+5=(n+2)^2+k^2+1$.
Indeed, a rational solution of the equation is given by
\[
  \frac{k^2+n^2}{\left(k^2+n^2+1\right) \left((k+1)^2+n^2+1\right)
    \left(k^2+(n+1)^2+1\right)}.
\]
\end{example}

\begin{example}\label{ex:5}
  Theorem~\ref{Thm:UniversalBound} is not sufficient for predicting periodic
  factors of a denominator. As an example, consider the equation
\begin{alignat*}1
 & 2 (k+n+1) f(n,k)-(k+3 n+8) f(n,k+1)\\
 &-(5 k+3 n+12) f(n+1,k)+3 (k+n+5) f(n+1,k+1)\\
 &+(k+n+5)f(n+2,k)=0.
\end{alignat*}
Possible choices for corner points are $(0,0),\penalty0(0,1),\penalty0(1,1)$ and
$(2,0)$.  Because of
\[
  \spread(k+n+1,3n+k+8)=\{\},
\]
one might be tempted to believe that only trivial denominators can occur in a
solution.  However, the equation admits the nontrivial rational function
solution
\[
  \frac{n^2+k^2}{(k+n+1)(k+n+2)(k+n+3)}.
\]
Observe that although not \textbf{every} corner point contains a shifted copy of
$k+n+1$, there is still \textbf{some} corner point which does.

This need not be the case, as indicated by the example
\[
  f(n+1,k)-f(n,k+1)=0
\]
already mentioned in the introduction.  This example, however, is special
because the shift set $\{(0,1),(1,0)\}$ belongs to a proper affine subspace
of~$\set Z^r$, and whenever this is the case, say the shift set belongs to a
subspace $L\subsetneq\set Z^r$, then for every vector $v=(v_1,\dots,v_r)\in
L^\bot\setminus\{0\}$ the polynomial $p=v_1n_1+\cdots+v_rn_r$ has the property
that $f$ is a rational solution of the equation if and only if $fp^\alpha$ is,
for any $\alpha\in\set Z$.  In particular, if there exists a rational solution
at all, then there also exists one whose denominator does not contain~$p$.

Apart from this exceptional situation, we observed on all the examples we
considered that periodic factors of a denominator appeared (possibly as a
shifted copy) in at least one of the coefficients corresponding to some corner
point of the shift set.  This at least suggests the periodic factors of these
coefficients as plausible guesses for the periodic part of the denominator
bound.
\end{example}

\section{Conclusion}

There are polynomials in several variables which may have an infinite spread.
Such polynomials are called periodic, while polynomials that cannot have
infinite spread are called aperiodic. Each polynomial can be split into a
periodic and an aperiodic part.

We have shown that for partial linear difference equations with polynomial
coefficients it is possible to determine all the factors that may possibly occur
in the aperiodic part of the denominator of a rational function solution. The
construction is a generalization of the corresponding result for univariate
equations.  It probably admits further generalization to $q$-equations or
equations whose coefficients belong to a $\Pi\Sigma$-field.
As it stands, Theorem~\ref{Thm:UniversalBound} will tend to produce
only a rough bound for the aperiodic part of the denominator, but we have
pointed out several refinements for improving the efficiency of the computation
on concrete examples.  Eventually, it would be interesting to see whether it is
possible to come up with an Abramov-style algorithm for directly computing the
greatest common divisor of all the individual bounds obtained from each corner
point. Until now, we have not succeeded in constructing such an algorithm.

It also remains open how to bound periodic factors of the denominator.  The
situation illustrated in Example~\ref{ex:5}, which seems to be typical both for
random equations and for equations coming from applications, indicates that our
result for aperiodic factors does not directly extend to periodic ones.  On the
other hand, it also indicates that an equation typically provides some hints for
the periodic factors in the denominators of its rational solutions.  This is
useful for making plausible heuristic guesses.  It also gives some hope that at
least for certain types of equations the periodic part of a denominator can be
found algorithmically. This needs further investigation.

\bibliographystyle{plain}
\bibliography{all}

\begin{thebibliography}{10}

\bibitem{abramov71}
Sergei~A. Abramov.
\newblock On the summation of rational functions.
\newblock {\em Zh. vychisl. mat. Fiz}, pages 1071--1075, 1971.

\bibitem{Abramov:89b}
Sergei~A. Abramov.
\newblock Problems in computer algebra that are connected with a search for
  polynomial solutions of linear differential and difference equations.
\newblock {\em Moscow Univ. Comput. Math. Cybernet.}, 3:63--68, 1989.

\bibitem{abramov95}
Sergei~A. Abramov.
\newblock Rational solutions of linear difference and $q$-difference equations
  with polynomial coefficients.
\newblock In {\em Proc. ISSAC'95}, July 1995.

\bibitem{abramov91}
Sergei~A. Abramov and K.Yu. Kvashenko.
\newblock Fast algorithms to search for the rational solutions of linear
  differential equations with polynomial coefficients.
\newblock In {\em Proc. ISSAC'91}, pages 267--270, 1991.

\bibitem{barkatou99}
Moulay Barkatou.
\newblock Rational solutions of matrix difference equations: The problem of
  equivalence and factorization.
\newblock In {\em Proc. ISSAC'99}, pages 277--282, 1999.

\bibitem{bostan06}
Alin Bostan, Frederic Chyzak, Thomas Cluzeau, and Bruno Salvy.
\newblock Low complexity algorithms for linear recurrences.
\newblock In Jean-Guillaume Dumas, editor, {\em Proc. ISSAC'06}, pages 31--39,
  2006.

\bibitem{bronstein00}
Manuel Bronstein.
\newblock On solutions of linear ordinary difference equations in their
  coefficient field.
\newblock {\em Journal of Symbolic Computation}, 29:841--877, 2000.

\bibitem{chen08b}
William Y.~C. Chen, Peter Paule, and Husam~L. Saad.
\newblock Converging to {G}osper's algorithm.
\newblock {\em Adv. in Appl. Math.}, 41(3):351--364, 2008.

\bibitem{abramov10}
Amel Gheffar and Sergei Abramov.
\newblock Valuations of rational solutions of linear difference equations at
  irreducible polynomials.
\newblock {\em (submitted)}.

\bibitem{Lang87}
Serge Lang.
\newblock {\em Linear Algebra}.
\newblock Springer, 1987.

\bibitem{Petkov:92}
Marko Petkov{\v s}ek.
\newblock Hypergeometric solutions of linear recurrences with polynomial
  coefficients.
\newblock {\em Journal of Symbolic Computation}, 14(2-3):243--264, 1992.

\bibitem{schneider04c}
Carsten Schneider.
\newblock A collection of denominator bounds to solve parameterized linear
  difference equations in {$\Pi\Sigma$}-extensions.
\newblock In {\em Proc. SYNASC'04}, pages 269--282, 2004.

\bibitem{hoeij98}
Mark van Hoeij.
\newblock Rational solutions of linear difference equations.
\newblock In {\em Proc. ISSAC'98}, pages 120--123, 1998.

\bibitem{weixlbaumer01}
Christian Weixlbaumer.
\newblock Solutions of difference equations with polynomial coefficients.
\newblock Master's thesis, RISC-Linz, 2001.

\end{thebibliography}

\end{document}